\documentclass[11pt]{article}
\usepackage{amssymb}
\usepackage[perpage,symbol]{footmisc}
\usepackage{listings}
\usepackage{amsfonts}
\usepackage{enumerate}
\usepackage{amsmath}
\usepackage{cite}
\usepackage{array}
\usepackage{booktabs}

\begin{document}

\newtheorem{theorem}{Theorem}[section]
\newtheorem{corollary}[theorem]{Corollary}
\newtheorem{definition}[theorem]{Definition}
\newtheorem{proposition}[theorem]{Proposition}
\newtheorem{conjecture}[theorem]{Conjecture}
\newtheorem{lemma}[theorem]{Lemma}
\newtheorem{example}[theorem]{Example}
\newenvironment{proof}{\noindent {\bf Proof.}}{\rule{3mm}{3mm}\par\medskip}
\newcommand{\remark}{\medskip\par\noindent {\bf Remark.~~}}
\title{Gold type codes of higher relative dimension}
\author{Chunlei Liu\footnote{Dept. of math., Shanghai Jiao Tong Univ., Sahnghai 200240, China, clliu@sjtu.edu.cn.}\ \footnote{Dengbi Technologies Cooperation Limited., Yichun 336099, China, 714232747@qq.com.}}
\date{}
\maketitle
\thispagestyle{empty}

\abstract{
Let $m,d,e,k$ be fixed positive integers such that
\[e=(m,d)=(m,2d), ~2\leq k \leq \frac{m+e}{2e}.\]
Let $s$ be a fixed maximum-length binary sequence of length $2^{m}-1$. Let
$(s_1,s_2,\cdots,s_{k-1})$ be a system of circular decimations of $s$ whose decimation factors are respectively
\[2^{d}+1,2^{2d}+1,\cdots,2^{(k-1)d}+1,\]
or respectively
\[2^{d}+1,2^{3d}+1,\cdots, 2^{(2k-3)d}+1,\]
or respectively
\[2^{(\frac{m-e}{2e})d}+1,2^{(\frac{m-3e}{2e})d}+1,\cdots,2^{(\frac{m+3e}{2e}-k)d}+1.\]
Then $s_1,\cdots,s_{k-1}$ are maximum-length binary sequences of length $2^{m}-1$. Let $C$ be the ${\mathbb F}_2$-vector space generated by all circular shifts of $s,s_1,\cdots,s_{k-1}$. Then $C$ has an ${\mathbb F}_{2^m}$-vector space structure, and is of dimension $k$ over ${\mathbb F}_{2^m}$. When $k=2$, $C$ is the Gold code. So we regard $C$ as a Gold type code of relative dimension $k$. The DC component distribution of $C$ is explicitly calculated out in the present paper.

\noindent {\bf Key phrases}: Gold code, cyclic code, alternating form


\section{\small{INTRODUCTION}}
\paragraph{}
Let $q$ be a prime power, and $C$ an $[n,k]$-linear code over ${\mathbb F}_q$. The weight of a codeword $c=(c_0,c_1,\cdots,c_{n-1})$ of $C$ is defined to be
\[{\rm wt}(c)=\#\{0\leq i\leq n-1|~c_i\neq0\}.\]
For each $i=0,1,\cdots,n$, define
\[A_i=\#\{c\in C\mid~{\rm wt}(c)=i\}.\]The sequence $(A_0,A_1,\cdots,A_n)$ is called the weight distribution of $C$.
Given a linear code $C$, it is challenging to determine its weight distribution.
The
weight distribution of Gold codes was determined by Gold \cite{Gold66, Gold67, Gold68}.
The
weight distribution of Kasami codes was determined by Kasami \cite{Kasami66}.
The
weight enumerators of Gold type and Kasami type codes of higher relative dimension were determined by Berlekamp \cite{Ber} and Kasami \cite{Kasami71}.
The
weight distribution of the $p$-ary analogue of Gold codes was determined by Trachtenberg \cite{Tr}.
The
weight distribution of the circular decimation of the $p$-ary analogue of Gold codes with decimation factor $2$ was determined by Feng-Luo \cite{FL}.
The
weight distribution of the $p$-ary analogue of Gold type codes of relative dimension $3$ was determined by Zhou-Ding-Luo-Zhang \cite{ZDLZ}.
The
weight distribution of the circular decimation with decimation factor $2$ of the $p$-ary analogue of Gold type codes  of relative dimension $3$ was determined by Zheng-Wang-Hu-Zeng \cite{ZWHZ}.
The
weight distribution of (the $p$-ary analogue of) Kasami type codes of maximum relative dimension was determined by Li-Hu-Feng-Ge \cite{LHFG}.
The
weight distribution of the $p$-ary analogue of Gold type codes of higher relative dimension was determined by Schmidt \cite{Sch}. The weight distribution of some other classes of cyclic codes was determined in the papers \cite{AL}, \cite{BEW}, \cite{BMC}, \cite{BMC10}, \cite{BMY}, \cite{De}, \cite{DLMZ}, \cite{DY}, \cite{FE}, \cite{FM}, \cite{KL}, \cite{LF}, \cite{LHFG}, \cite{LN}, \cite{LYL}, \cite{LTW}, \cite{MCE}, \cite{MCG}, \cite{MO}, \cite{MR}, \cite{MY}, \cite{MZLF}, \cite{RP}, \cite{SC}, \cite{VE}, \cite{WTQYX}, \cite{XI}, \cite{XI12}, \cite{YCD}, \cite{YXDL} and \cite{ZHJYC}.
\paragraph{}
Let $m,d,e,k$ be fixed positive integers such that
\[e=(m,d)=(m,2d), ~2\leq k \leq \frac{m+e}{2e}.\]
Let $s$ be a fixed maximum-length binary sequence of length $2^{m}-1$. Let
$(s_1,s_2,\cdots,s_{k-1})$ be a system of circular decimations of $s$ whose decimation factors are respectively
\[2^{d}+1,2^{2d}+1,\cdots,2^{(k-1)d}+1,\]
or respectively
\[2^{d}+1,2^{3d}+1,\cdots, 2^{(2k-3)d}+1,\]
or respectively
\[2^{(\frac{m-e}{2e})d}+1,2^{(\frac{m-3e}{2e})d}+1,\cdots,2^{(\frac{m+3e}{2e}-k)d}+1.\]
Then $s_1,\cdots,s_{k-1}$ are maximum-length binary sequences of length $2^{m}-1$. Let $C$ be the ${\mathbb F}_2$-vector space generated by all circular shifts of $s,s_1,\cdots,s_{k-1}$. If $d=e=1$, then $C$ is the code studied by Berlekamp \cite{Ber} and Kasami \cite{Kasami71}.
Let $\{Q_{\vec a}\}$ be the system
\[
Q_{\vec {a}}(x)={\rm  Tr}_{\mathbb{F}_{2^{m}}/\mathbb{F}_{2^e}}(a_0x)+\sum_{j=1}^{k-1}{\rm  Tr}_{\mathbb{F}_{2^{m}}/\mathbb{F}_{2^e}}(a_jx^{2^{jd}+1}),~\vec {a}\in\mathbb{F}_{2^{m}}^{k},
\]
or the system
\[
Q_{\vec {a}}(x)={\rm  Tr}_{\mathbb{F}_{2^{m}}/\mathbb{F}_{2^e}}(a_0x)+\sum_{j=1}^{k-1}{\rm  Tr}_{\mathbb{F}_{2^{m}}/\mathbb{F}_{2^e}}(a_jx^{2^{(2j-1)d}+1}),~\vec {a}\in\mathbb{F}_{2^{m}}^{k},
\]
or the system
\[
Q_{\vec {a}}(x)={\rm  Tr}_{\mathbb{F}_{2^{m}}/\mathbb{F}_{2^e}}(a_0x)+\sum_{j=1}^{k-1}{\rm  Tr}_{\mathbb{F}_{2^{m}}/\mathbb{F}_{2^e}}(a_jx^{2^{(\frac{m+e}{2e}-j)d}+1}),~\vec {a}\in\mathbb{F}_{2^{m}}^{k}.
\]
Then \[C=\{c_{\vec{a}}\mid~\vec{a}\in{\mathbb F}_{2^m}^k\},\]
where $c_{\vec a}=({\rm Tr}_{{\mathbb F}_{2^e}/{\mathbb F}_2}(Q_{\vec a} (\pi^{-i}))_{i=0}^{2^m-2}$ with $\pi$ being a primitive element of ${\mathbb F}_{2^m}$. The correspondence $\vec{a}\mapsto c_{\vec a}$ defines an ${\mathbb F}_{2^m}$-vector space structure on $C$, and $C$ is of dimension $k$ over ${\mathbb F}_{2^m}$. When $k=2$, $C$ is the Gold code. So we call $C$ a Gold type code of relative dimension $k$.
\paragraph{}
One can prove the following.
\begin{theorem}\label{dcbound}
If $c\in C$ is nonzero, then \[{\rm DC}(c)\in\{-1,-1+\pm2^{\frac{m+e}{2}+je}\mid~j=0,1,2,\cdots,k-2\},\]
where
\[{\rm DC}(c)=2^m-1-2{\rm wt}(c)=\sum_{i=0}^{2^m-2}(-1)^{c_i}\]
is the DC component of $c=(c_0,c_1,\cdots,c_{2^m-2})\in C$.
\end{theorem}
The present paper is concerned with the frequencies
\begin{equation}\label{dcfrequencydef}\alpha_{r,\varepsilon}=\#\{0\neq c\in C\mid~{\rm DC}(c)=-1+\varepsilon 2^{m-\frac{er}{2}}\},~r=0,2,4,\cdots,\frac{m-e}{e}.\end{equation}
The main result of the present paper is the following.
\begin{theorem}\label{main}
For each $j=0,1,\cdots,k-2$, and for each $\varepsilon=\pm1$, we have
\[\alpha_{\frac{m-e}{e}-2i,\varepsilon}=\frac12(2^{m-e-2ei}+\varepsilon2^{\frac{m-e}{2}-ei})\sum_{j=i}^{k-2}(-1)^{j-i}4^{e\binom{j-i}{2}}
\binom{j}{i}_{4^{e}}\binom{\frac{m-e}{2e}}{j}_{4^e}(2^{m(k-1-j)}-1),\]
where $\binom{j}{i}_{q}$ is a Gaussian binomial coefficient.
\end{theorem}
From the above theorem one can deduce the following.
\begin{theorem}\label{balanced}We have
\[\begin{split}&\#\{c\in C\mid
~{\rm DC}(c)=-1\}\\&=2^{mk}-1-\sum_{u=0}^{k-2}(-1)^u2^{-e(u+1)^2}(2^{m(k-u)}-2^m)\prod_{j=0}^{u-1}(2^m-2^{e(2j+1)})
\\&\approx2^{mk}(1-\sum_{u=0}^{k-2}(-1)^u2^{-e(u+1)^2}).\end{split}\]
\end{theorem}
If $d=e=1$, then the weight enumerator of $C$ is determined by Berlekamp \cite{Ber} and Kasami \cite{Kasami71}. However, some extra calculations are needed to explicitly write out the coefficients of the weight enumerators in \cite{Ber, Kasami71}.
\section{\small{ENTERING BILINEAR FORMS I}}
 \paragraph{}
In this section we shall prove Theorem \ref{dcbound}.
 \paragraph{}
Note that
\begin{equation}\label{dcexpsumrelation}1+{\rm DC}(c_{\vec a})=\sum_{x\in{\mathbb F}_{2^m}}(-1)^{{\rm Tr}_{{\mathbb F}_{2^e}/{\mathbb F}_2}(Q_{\vec a} (x))}.\end{equation}
It is well-known that
\begin{equation}\label{quadraticexpsum}
\sum_{x\in{\mathbb F}_{2^m}}(-1)^{{\rm Tr}_{{\mathbb F}_{2^e}/{\mathbb F}_2}(Q_{\vec a} (x))}=\begin{cases}
           0,&  2 \nmid {\rm rk}(Q_{\vec a}),\\
             \pm2^{m-e\cdot\frac{{\rm rk}(Q_{\vec a})}{2})},&2 |{\rm rk}(Q_{\vec a}).
           \end{cases}
\end{equation}
 \paragraph{}
Let \[
B_{\vec {a}}(x,y)=Q_{\vec a}(x+y)-Q_{\vec a}(x)-Q_{\vec a}(y).\]
Then
$\{B_{\vec a}\}$ is either the system
\begin{equation}\label{bilinearform1}
B_{\vec {a}}(x,y)=\sum_{j=1}^{k-1}{\rm  Tr}_{\mathbb{F}_{2^{m}}/\mathbb{F}_{2^e}}(a_{j}(xy^{2^{jd}}+x^{2^{jd}}y)),~{\vec a}\in{\mathbb F}_{2^m}^k,\end{equation}
or the system
\begin{equation}\label{bilinearform2}
B_{\vec {a}}(x,y)=\sum_{j=1}^{k-1}{\rm  Tr}_{\mathbb{F}_{2^{m}}/\mathbb{F}_{2^e}}(a_{j}(xy^{2^{(2j-1)d}}+x^{2^{(2j-1)d}}y)),~{\vec a}\in{\mathbb F}_{2^m}^k,\end{equation}
or the system
\begin{equation}\label{bilinearform3}
B_{\vec {a}}(x,y)=\sum_{j=1}^{k-1}{\rm  Tr}_{\mathbb{F}_{2^{m}}/\mathbb{F}_{2^e}}(a_{j}(xy^{2^{(\frac{m+e}{2e}-j)d}}+x^{2^{(\frac{m+e}{2e}-j)d}}y)),~{\vec a}\in{\mathbb F}_{2^m}^k.\end{equation}
It is well-known that
\begin{equation}\label{rkbilinear}{\rm rk}(B_{\vec a})=\left\{
                             \begin{array}{ll}
                              {\rm rk}(Q_{\vec a}), & \hbox{} 2\mid{\rm rk}(Q_{\vec a}),\\
                               {\rm rk}(Q_{\vec a})-1, & \hbox{}2\nmid{\rm rk}(Q_{\vec a}).
                             \end{array}
                           \right.\end{equation}
 \paragraph{}
We now prove Theorem \ref{dcbound}.
By (\ref{dcexpsumrelation}), (\ref{quadraticexpsum}) and (\ref{rkbilinear}), it suffices to prove the following.
\begin{theorem}\label{rankbound}If $(a_1,\cdots,a_{k-1})\neq0$, then
\[{\rm rk}(B_{\vec a})\geq \frac{m-e}{e}-2(k-2).\]
\end{theorem}
\begin{proof} Suppose that $(a_1,\cdots,a_{k-1})\neq 0$. It suffices to show that
\[
 {\rm dim}_{\mathbb{F}_{2^e}}{\rm Rad}({B_{\vec {a}}} )\leq 2(k-1),
 \]
where
\[{\rm Rad}({B_{\vec {a}}})=\{x \in  \mathbb{F}_{2^{m}}\mid~B_{\vec a}(x,y)=0,~\forall y\in{\mathbb F}_{2^m}\}.\]
Without loss of generality, we assume that $\{B_{\vec a}\}$ is the system (\ref{bilinearform1}). Then
\[\begin{split}
{\rm Rad}({B_{\vec {a}}})&=\{x \in  \mathbb{F}_{2^{m}}|\sum_{j=1}^{k-1}(a_{j}^{2^{-jd}}x^{2^{-jd}}+a_{j}x^{2^{jd}})=0\}\\
&=\{x \in  \mathbb{F}_{2^{m}}|\sum_{j=1}^{k-1}(a_{j}^{2^{(k-1-j)d}}x^{2^{(k-1-j)d}}+a_{j}^{2^{(k-1)d}}x^{2^{(k-1+j)d}})=0\}.
\end{split}
\]Note that
\[
 \{x \in  \mathbb{F}_{2^{md/e}}|\sum_{j=1}^{k-1}(a_{j}^{2^{(k-1-j)d}}x^{2^{(k-1-j)d}}+a_{j}^{2^{(k-1)d}}x^{2^{(k-1+j)d}})=0\}.
 \]
is a subspace of $\mathbb{F}_{2^{md/e}}$ over ${\mathbb F}_{2^d}$ of dimension $\leq 2(k-1)$.
As $(m, d)=e$,
a basis of $\mathbb{F}_{2^{m}}$ over ${\mathbb{F}_{2^{e}}}$ is also a basis of $\mathbb{F}_{2^{md/e}}$ over $\mathbb{F}_{2^{d}}$.
It follows that
\[
 {\rm dim}_{\mathbb{F}_{2^e}}{\rm Rad}({B_{\vec {a}}})\leq 2(k-1).
 \]
The theorem is proved.
\end{proof}
\section{\small{ENTERING BILINEAR EQUATIONS II}}
In this section we shall reduce Theorem \ref{main} to  the following.
\begin{theorem}\label{bilinearrkdis}We have, for $0\leq i\leq k-2$,
\[\beta_{\frac{m-e}{e}-2i}=\sum_{j=i}^{k-2}(-1)^{j-i}2^{e(j-i)(j-i-1)}
\binom{j}{i}_{4^{e}}\binom{\frac{m-e}{2e}}{j}_{4^e}(2^{m(k-1-j)}-1),\]
where
\begin{equation}\label{rkfrequencydef}\beta_{r}=2^{-m}\#\{\vec{a}\in{\mathbb F}_{2^m}^k\mid~{\rm rk}(B_{\vec a})=r,~(a_1,\cdots,a_{k-1})\neq0\}.\end{equation}
\end{theorem}
It suffices to prove the following.
\begin{theorem}\label{signeffect}For each $r=0,2,\cdots,\frac{m-e}{e}$, \[\alpha_{r,\varepsilon} =\frac12(2^{er}+\varepsilon2^{\frac{er}{2}})\beta_r,\]
\end{theorem}
\begin{proof} By (\ref{dcfrequencydef}), (\ref{dcexpsumrelation}), (\ref{quadraticexpsum}), (\ref{rkbilinear}), and (\ref{rkfrequencydef}), \[\begin{split}&2^{m-\frac{er}{2}}(\alpha_{r,1} -\alpha_{r,-1})\\&=
\sum_{{\rm rk}(B_{\vec a} )=r}\sum_{x\in{\mathbb F}_{2^m}}(-1)^{{\rm Tr}_{{\mathbb F}_{2^e}/{\mathbb F}_2}(Q_{\vec a} (x))}\\&=2^{-m}\sum_{c\in{\mathbb F}_{2^m}}\sum_{{\rm rk}(B_{\vec a} )=r}\sum_{x\in{\mathbb F}_{2^m}}(-1)^{{\rm Tr}_{{\mathbb F}_{2^e}/{\mathbb F}_2}({\rm Tr}_{{\mathbb F}_{2^m}/{\mathbb F}_{2^e}}(cx)+Q_{\vec a} (x))}\\&=2^{-m}\sum_{{\rm rk}(B_{\vec a} )=r}\sum_{x\in{\mathbb F}_{2^m}}(-1)^{{\rm Tr}_{{\mathbb F}_{2^e}/{\mathbb F}_2}(Q_{\vec a} (x))}\sum_{c\in{\mathbb F}_{2^m}}(-1)^{{\rm Tr}_{{\mathbb F}_{2^m}/{\mathbb F}_{2}}(cx)}\\
&=2^{m}
\beta_r.\end{split}
\]
Similarly,
 \[\begin{split}&2^{2m-er}(\alpha_{r,1} +\alpha_{r,-1})\\&=
\sum_{{\rm rk}(B_{\vec a} )=r}(\sum_{x\in{\mathbb F}_{2^m}}(-1)^{{\rm Tr}_{{\mathbb F}_{2^e}/{\mathbb F}_2}(Q_{\vec a} (x))})^2\\&=2^{-m}\sum_{c\in{\mathbb F}_{2^m}}\sum_{{\rm rk}(B_{\vec a} )=r}(\sum_{x\in{\mathbb F}_{2^m}}(-1)^{{\rm Tr}_{{\mathbb F}_{2^e}/{\mathbb F}_2}({\rm Tr}_{{\mathbb F}_{2^m}/{\mathbb F}_{2^e}}(cx)+Q_{\vec a} (x))})^2\\&=2^{-m}\sum_{{\rm rk}(B_{\vec a} )=r}\sum_{x,y\in{\mathbb F}_{2^m}}(-1)^{{\rm Tr}_{{\mathbb F}_{2^e}/{\mathbb F}_2}(Q_{\vec a} (x)+Q_{\vec a} (y))}\sum_{c\in{\mathbb F}_{2^m}}(-1)^{{\rm Tr}_{{\mathbb F}_{2^m}/{\mathbb F}_{2}}(c(x+y))}\\
&=2^{2m}
\beta_r.\end{split}
\]
The theorem is proved.\end{proof}
\section{\small{ASSOCIATION SCHEME THEORETIC APPROACH}}
In this section we shall
use the following theorem of Delarte-Goethals to prove Theorem \ref{bilinearrkdis}.
\begin{theorem}[\cite{DG}] Let $M$ be an odd number, $X$ the space of alternating bilinear forms on an $M$-dimension vector space over ${\mathbb F}_{q}$, $Y$ a subspace of $X$, and
\[d(Y)=\min\{{\rm rk}(y)\mid 0\neq y\in Y\}.\]
Then
\[|Y|\leq q^{M(M-d(Y)+1)/2}.\]
Moreover, if the equality holds, then, for $i\leq (M-1-d(Y))/2$,
\[\begin{split}&\#\{y\in Y\mid~{\rm rk}(y)=M-1-2i\}\\&=\sum_{j=i}^{(M-1-d(Y))/2}(-1)^{j-i}q^{(j-i)(j-i-1)}
\binom{j}{i}_{q^2}\binom{(M-1)/2}{j}_{q^2}(q^{M(M-d(Y)+1-2j)/2}-1).\end{split}\]
\end{theorem}
We now the above theorem to prove Theorem \ref{bilinearrkdis}.\paragraph{}
Let $X$ be the space of alternating ${\mathbb F}_{2^e}$-bilinear forms on ${\mathbb F}_{2^m}$. Fix a system $\{B_{\vec a}\}$. Set
\[Y=\{B_{\vec a}\mid~\vec{a}\in{\mathbb F}_{2^m}^{k},a_0=0\}.\]
By Theorem \ref{rankbound}, \[d(Y)\geq \frac{m-e}{e}-2(k-2).\]
By Delsarte-Goethals' theorem,
\[|Y|\leq 2^{m(\frac{m+e}{e}-d(Y))/2}\leq 2^{m(k-1)}.\]
As $|Y|=2^{m(k-1)}$, we arrive at
\[|Y|=2^{m(\frac{m+e}{e}-d(Y))/2}= 2^{m(k-1)}.\]
In particular,
$d(Y)=\frac{m-e}{e}-2(k-2)$.
Applying Delsarte-Goethals' theorem one more time, we have, for $0\leq i\leq k-2$,
\[\begin{split}&\#\{\vec{a}\in {\mathbb F}_{2^m}^{k}\mid~{\rm rk}(B_{\vec a})=\frac{m-e}{e}-2i,a_0=0\}\\&=\sum_{j=i}^{k-2}(-1)^{j-i}2^{e(j-i)(j-i-1)}
\binom{j}{i}_{4^{e}}\binom{\frac{m-e}{2e}}{j}_{4^e}(2^{m(k-1-j)}-1).\end{split}\]
Theorem \ref{bilinearrkdis} is proved.

\section{\small{NUMBER THEORETIC APPROACH}}
\paragraph{}
The theorem of Delarte-Goethals we used in the last section is proved by developing the theory of association schemes.
To make the present paper self-contained, we shall develop a number theoretic approach, which is similar to the approach of Berlekamp \cite{Ber} and Kasami \cite{Kasami71}.\paragraph{}
Let $V_{s,u}$ be the set of solutions
$(x_1,x_2,\cdots,x_{2u})\in{\mathbb F}_{2^m}^{2u}$ of one of the systems
\begin{equation}\label{bilinearminimalnumber1}
\sum_{i=1}^u(x_{2i-1}x_{2i}^{2^{jd}}+x_{2i-1}^{2^{jd}}x_{2i})=0,~j=1,2,\cdots,s,
\end{equation}
\begin{equation}\label{bilinearminimalnumber2}
\sum_{i=1}^u(x_{2i-1}x_{2i}^{2^{(2j-1)d}}+x_{2i-1}^{2^{(2j-1)d}}x_{2i})=0,~j=1,2,\cdots,s,
\end{equation}
and
\begin{equation}\label{bilinearminimalnumber3}
\sum_{i=1}^u(x_{2i-1}x_{2i}^{2^{(\frac{m+e}{2e}-j)d}}+x_{2i-1}^{2^{(\frac{m+e}{2e}-j)d}}x_{2i})=0,~j=1,2,\cdots,s.
\end{equation}
In this section we shall use the following theorem to prove Theorem \ref{bilinearrkdis}.
\begin{theorem}\label{lineardependence}
If $s\geq u\geq1$, then $V_{s,u}=V_{u,u}$.
\end{theorem}
We now prove Theorem \ref{bilinearrkdis}.
We shall make repeated use of the following $q$-binomial formula
\[\prod_{i=0}^{u-1}(1+q^it)=\sum_{i=0}^uq^{\binom{i}{2}}\binom{u}{i}_qt^{i}.\]
By the orthogonality of characters and Theorem \ref{lineardependence}, we have
\[
\sum_{{\vec a}\in\mathbb{F}_{2^{m}}^{k}}\big(\sum_{x,y\in{\mathbb F}_{2^m}}(-1)^{{\rm Tr}_{{\mathbb F}_{2^e}/{\mathbb F}_2}(B_{\vec a}(x,y))}\big)^u=2^{mk}|V_{u,u}|,~0\leq u\leq k-1,
\]
where $|V_{0,0}|=1$.
Applying the identity
\[\sum_{x,y\in{\mathbb F}_{2^m}}(-1)^{{\rm Tr}_{{\mathbb F}_{2^e}/{\mathbb F}_2}(B_{\vec a}(x,y))}=2^{2m-e\cdot{\rm rk}(B_{\vec a})},
\]
we arrive at \[\sum_{2\mid r=\frac{m-e}{e}-2(k-2)}^{\frac{m-e}{e}}\beta_r2^{u(2m-er)}=2^{m(k-1)}|V_{u,u}|-2^{2mu},~0\leq u\leq k-1.\]
That is,
\[\sum_{i=0}^{k-2}\beta_{\frac{m-e}{e}-2i}4^{eui}=2^{m(k-1)-(m+e)u}|V_{u,u}|-2^{(m-e)u},
~0\leq u\leq k-1.\]
Consider the equation
\[\sum_{i=0}^{k-2}\beta_{\frac{m-e}{e}-2i}\left(
                                            \begin{array}{c}
                                              1 \\
                                              4^{ei} \\
                                              \vdots \\
                                              4^{e(k-1)i} \\
                                            \end{array}
                                          \right)
=\left(
   \begin{array}{c}
      2^{m(k-1)}|V_{0,0}|-1 \\
     2^{m(k-1)-(m+e)}|V_{1,1}|-2^{m-e} \\
     \vdots \\
     2^{-(k-1)e}|V_{k-1,k-1}|-2^{u(m-e)} \\
   \end{array}
 \right)\]
Multiplying on the left by the row vector $((-1)^{k-1-i}4^{e\binom{k-1-i}{2}}\binom{k-1}{i}_{4^e})_{i=0}^{k-1}$, and applying the $q$-binomial formula, we arrive at
\[\sum_{i=0}^{k-1}(-1)^{k-1-i}4^{e\binom{k-1-i}{2}}\binom{k-1}{i}_{4^e}(2^{m(k-1)-(m+e)i}|V_{i,i}|-2^{(m-e)i})=0.\] Applying the $q$-binomial formula once more, we arrive at
\[\sum_{i=0}^{k-1}(-1)^{k-1-i}4^{e\binom{k-1-i}{2}}\binom{k-1}{i}_{4^e}2^{m(k-1)-(m+e)i}|V_{i,i}|
=\prod_{i=0}^{k-1}(2^{(m-e)i}-4^{ei}).\]
Replacing $k-1$ with an arbitrary positive integer $u$, we arrive at
\[\sum_{i=0}^{u}(-1)^{u-i}4^{e\binom{u-i}{2}}\binom{u}{i}_{4^e}2^{mu-(m+e)i}|V_{i,i}|
=\prod_{i=0}^{k-1}(2^{(m-e)i}-4^{ei}).\]
That is,
\begin{equation}\label{recursiverelation}\sum_{i=0}^{u}(-1)^{u-i}4^{e\binom{u-i}{2}}\binom{u}{i}_{4^e}2^{-(m+e)i}|V_{i,i}|
=2^{-mu}\prod_{i=0}^{k-1}(2^{(m-e)i}-4^{ei}).\end{equation}
Now fix $0\leq u\leq k-1$, and consider the equation
\[\sum_{i=0}^{k-2}\beta_{\frac{m-e}{e}-2i}\left(
                                            \begin{array}{c}
                                              1 \\
                                              4^{ei} \\
                                              \vdots \\
                                              4^{eui} \\
                                            \end{array}
                                          \right)
=\left(
   \begin{array}{c}
      2^{m(k-1)}|V_{0,0}|-1 \\
     2^{m(k-1)-(m+e)}|V_{1,1}|-2^{m-e} \\
     \vdots \\
     2^{-(k-1)e}|V_{u,u}|-2^{u(m-e)} \\
   \end{array}
 \right)\]
Multiplying on the left by the row vector $((-1)^{u-i}q^{\binom{u-i}{2}}\binom{u}{i}_{q})_{i=0}^u$, and applying the $q$-binomial formula as well as (\ref{recursiverelation}), we arrive at
\[\sum_{i=u}^{k-2}\beta_{\frac{m-e}{e}-2i}\prod_{0\leq h\leq u-1}(4^{ei}-4^{eh})
=(2^{m(k-1-u)}-1)\prod_{0\leq h\leq u-1}(2^{m-e}-4^{eh}).
\]
Dividing both sides by $\prod_{0\leq h\leq u-1}(4^{eu}-4^{eh})$, we arrive at
\[\sum_{i=u}^{k-2}\beta_{\frac{m-e}{e}-2i}\binom{i}{u}_{4^e}
=\binom{\frac{m-e}{2e}}{u}_{4^e}(2^{m(k-1-u)}-1).
\]
Applying the
$q$-binomial M\"{o}bius inversion formula
\[ \sum_{i=v}^u(-1)^{i-v}q^{\binom{i-v}{2}}\binom{i}{v}_{q}\binom{u}{i}_{q}=\left\{
                                                                              \begin{array}{ll}
                                                                                1, & \hbox{} u=v,\\
                                                                                0, & \hbox{}u\neq v,
                                                                              \end{array}
                                                                            \right.
\]
we arrive at
\[\beta_{\frac{m-e}{e}-2j}
=\sum_{u=j}^{k-2}(-1)^{u-j}4^{e\binom{u-j}{2}}\binom{\frac{m-e}{2e}}{u}_{4^e}\binom{u}{j}_{4^e}(2^{m(k-1-u)}-1).
\]
Theorem \ref{bilinearrkdis} is proved.
\section{\small{SYSTEMS OF BILINEAR EQUATIONS}}
In this section we shall prove Theorem \ref{lineardependence}.
We begin with the following.
\begin{theorem}\label{elimination}The systems (\ref{bilinearminimalnumber1}), (\ref{bilinearminimalnumber2}), and (\ref{bilinearminimalnumber3}) are respectively equivalent to the systems
\begin{equation}\label{bilinearminimalnumber1e}
\begin{cases}\sum_{i=1}^u(x_{2i-1}x_{2i}^{2^{d}}+x_{2i-1}^{2^{d}}x_{2i})=0,\\
\sum_{i=1}^{u}(\tilde{x}_{2i-1}\tilde{x}_{2i}^{2^{jd}}+\tilde{x}_{2i-1}^{2^{jd}}\tilde{x}_{2i})=0,\\
j=1,2,\cdots,s-1,\end{cases}\end{equation}
\begin{equation}\label{bilinearminimalnumber2e}
\begin{cases}\sum_{i=1}^u(x_{2i-1}x_{2i}^{2^{d}}+x_{2i-1}^{2^{d}}x_{2i})=0,\\
\sum_{i=1}^{u}(\tilde{x}_{2i-1}\tilde{x}_{2i}^{2^{(2j-1)d}}+\tilde{x}_{2i-1}^{2^{(2j-1)d}}\tilde{x}_{2i})=0,\\
j=1,2,\cdots,s-1,\end{cases}\end{equation}
and
\begin{equation}\label{bilinearminimalnumber3e}
\begin{cases}\sum_{i=1}^u(x_{2i-1}x_{2i}^{2^{(\frac{m-e}{2e})d}}+x_{2i-1}^{2^{(\frac{m-e}{2e})d}}x_{2i})=0,\\
\sum_{i=1}^{u}(\tilde{x}_{2i-1}\tilde{x}_{2i}^{2^{(\frac{m+e}{2e}-j)d}}+\tilde{x}_{2i-1}^{2^{(\frac{m+e}{2e}-j)d}}\tilde{x}_{2i})=0,\\
j=1,2,\cdots,s-1,\end{cases}\end{equation}
where $\tilde{x}_i=x_i+x_i^{2^d},~x_i+x_i^{2^{2d}}$, and $x_i+x_i^{2^{-d}}$ respectively.
\end{theorem}
\begin{proof} We deal with the system (\ref{bilinearminimalnumber1}) first.
Adding $2^d$-th power of the $(j-1)$-th equation to the $j$-th equation, we arrive at
\[\begin{cases}\sum_{i=1}^u(x_{2i-1}x_{2i}^{2^{d}}+x_{2i-1}^{2^{d}}x_{2i})=0,\\
\sum_{i=1}^u(x_{2i-1}x_{2i}^{2^{jd}}+x_{2i-1}^{2^{jd}}x_{2i}+x_{2i-1}^{2^d}x_{2i}^{2^{jd}}+x_{2i-1}^{2^{jd}}x_{2i}^{2^d})=0,\\
~j=2,3,\cdots,s.\end{cases}\]
Adding the $(j-1)$-th equation to the $j$-th equation in the above system, we arrive at the system
(\ref{bilinearminimalnumber1e}).
\paragraph{}We now deal with the system (\ref{bilinearminimalnumber2}).
Adding $2^{2d}$-th power of the $(j-1)$-th equation to the $j$-th equation, we arrive at
\[\begin{cases}\sum_{i=1}^u(x_{2i-1}x_{2i}^{2^{d}}+x_{2i-1}^{2^{d}}x_{2i})=0,\\
\sum_{i=1}^u(x_{2i-1}x_{2i}^{2^{(2j-1)d}}+x_{2i-1}^{2^{(2j-1)d}}x_{2i}+x_{2i-1}^{2^d}x_{2i}^{2^{(2j-1)d}}+x_{2i-1}^{2^{(2j-1)d}}x_{2i}^{2^d})=0,\\
~j=2,3,\cdots,s.\end{cases}\]
Adding the $(j-1)$-th equation to the $j$-th equation in the above system, we arrive at the system
(\ref{bilinearminimalnumber2e}).
 \paragraph{}Finally we deal with the system (\ref{bilinearminimalnumber3}).
Inserting $2^{\frac{m+e}{2e}}$-th power of the first equation to the system, we arrive at
 the system
\[
\sum_{i=1}^u(x_{2i-1}x_{2i}^{2^{(\frac{m+e}{2e}-j)d}}+x_{2i-1}^{2^{(\frac{m+e}{2e}-j)d}}x_{2i})=0,~j=0,1,2,\cdots,s.\]
Adding the $2^{-d}$-th power of the $(j-1)$-th equation to the $j$-th equation in the above system,
we arrive at
\[
\begin{cases}\sum_{i=1}^u(x_{2i-1}x_{2i}^{2^{(\frac{m-e}{2e})d}}+x_{2i-1}^{2^{(\frac{m-e}{2e})d}}x_{2i})=0,\\
\sum_{i=1}^u(x_{2i-1}x_{2i}^{2^{(\frac{m+e}{2e}-j)d}}+x_{2i-1}^{2^{(\frac{m+e}{2e}-j)d}}x_{2i}
+x_{2i-1}^{2^{-d}}x_{2i}^{2^{(\frac{m+e}{2e}-j)d}}+x_{2i-1}^{2^{(\frac{m+e}{2e}-j)d}}x_{2i}^{2^{-d}})=0,\\
j=1,2,\cdots,s.\end{cases}\]
Adding the $(j-1)$-th equation to the $j$-th equation in the above system, we arrive at the system (\ref{bilinearminimalnumber3e}).
Theorem \ref{elimination} is proved.
\end{proof}
We now  prove Theorem \ref{lineardependence}.
If $u=1$, then $V_{s,u}=V_{u,u}$ trivially. Now  assume that $u\geq 2$.
Suppose that $(x_{1}, x_2, \cdots, x_{2u})$ belongs to $V_{u,u}$. We are going to show that $(x_{1}, x_2, \cdots, x_{2u})$ belongs to $V_{s,u}$. By induction, we may assume that $x_{2u}\neq0$. Then we may further assume that $x_{2u}=1$. By Theorem \ref{elimination}, $(\tilde{x}_{1}, \tilde{x}_2, \cdots, \tilde{x}_{2u})\in V_{u-1,u}$.
As $\tilde{x}_{2u}=0$, we see that $(\tilde{x}_{1}, \tilde{x}_2, \cdots, \tilde{x}_{2u-2})\in V_{u-1,u-1}$. By induction, $(\tilde{x}_{1}, \tilde{x}_2, \cdots, \tilde{x}_{2u-2})\in V_{s-1,u-1}$. As $\tilde{x}_{2u}=0$, we see that $(\tilde{x}_{1}, \tilde{x}_2, \cdots, \tilde{x}_{2u})\in V_{s-1,u}$. By Theorem \ref{elimination}, $(x_{1}, x_2, \cdots, x_{2u})$ belongs to $V_{s,u}$. Theorem \ref{lineardependence} is proved.
\section{\small{THE NUMBER OF BALANCED SEQUENCES}}
In this section we prove Theorem \ref{balanced}. We have
\[\begin{split}&\#\{c\in C\mid
~{\rm DC}(c)=-1\}\\
=&2^{mk}-1-\sum_{j=0}^{\frac{m-e}{2e}}2^{m-e-2ej}
\sum_{u=j}^{k-2}(-1)^{u-j}4^{e\binom{u-j}{2}}\binom{\frac{m-e}{2e}}{u}_{4^e}\binom{u}{j}_{4^e}(2^{m(k-1-u)}-1)\\
=&2^{mk}-1-2^{m-e}\sum_{u=0}^{k-2}4^{-eu}\binom{\frac{m-e}{2e}}{u}_{4^e}(2^{m(k-1-u)}-1)
\sum_{j=0}^{u}(-1)^{j}4^{ej}4^{e\binom{j}{2}}\binom{u}{j}_{4^e}\\
=&2^{mk}-1-2^{m-e}\sum_{u=0}^{k-2}4^{-eu}\binom{\frac{m-e}{2e}}{u}_{4^e}(2^{m(k-1-u)}-1)
\prod_{j=1}^{u}(1-4^{ej})\\
=&2^{mk}-1-\sum_{u=0}^{k-2}(-1)^u2^{-e(u+1)^2}(2^{m(k-u)}-2^m)\prod_{j=0}^{u-1}(2^m-2^{e(2j+1)})\\
&\approx2^{mk}(1-\sum_{u=0}^{k-2}(-1)^u2^{-e(u+1)^2}).\end{split}\]
Theorem \ref{balanced} is proved.
\paragraph{}
{\bf Acknowledgement.} The author thanks Kai-Uwe Schmidt for telling him the background of this subject.

\end{document}